\documentclass[letterpaper, twocolumn, 10pt]{IEEEtran}

\usepackage{amsmath, amssymb}
\usepackage{amsthm}
\usepackage{algorithm}
\usepackage{algorithmic}
\usepackage{color, array, colortbl}
\usepackage{graphicx}
\usepackage{multirow}
\usepackage[margin=3.2cm]{geometry}
\usepackage{color, array, colortbl}
\usepackage{graphicx}
\usepackage{comment}
\usepackage{multirow}
\usepackage{framed,color}
\definecolor{shadecolor}{rgb}{0.9,0.9,0.9}
\newcommand{\mathsym}[1]{{}}

\newtheorem{theorem}{Theorem}[section]
\newtheorem{lemma}[theorem]{Lemma}
\newtheorem{claim}[theorem]{Claim}
\newtheorem{fact}[theorem]{Fact}

\newcommand{\cancel}[1]{}

\newcommand{\E}{\ensuremath{\mathbb{E}}}
\renewcommand{\Pr}{\ensuremath{\mathbb{P}}}

\begin{document}

\title{AntiJam: Efficient Medium Access despite\\Adaptive and Reactive Jamming}

\author {
   Andrea Richa$^1$, Christian Scheideler$^2$, Stefan Schmid$^3$, Jin Zhang$^1$\\
   \small $^1$ Computer Science and Engineering, SCIDSE,  Arizona State University\\
       \small Tempe, Arizona, USA; \{aricha,jzhang82\}@asu.edu\\
  \small $^2$ Department of Computer Science, University of
Paderborn, D-33102 Paderborn,
  Germany\\ \small scheideler@upb.de\\
 \small $^3$ Deutsche Telekom Laboratories, TU Berlin, D-10587 Berlin,
  Germany\\ \small stefan@net.t-labs.tu-berlin.de\
}

%\date{}

\maketitle

\begin{abstract}
Intentional interference constitutes a major threat for
communication networks operating over a shared medium where
availability is imperative. Jamming attacks are often simple and
cheap to implement. In particular, today's jammers can perform
physical carrier sensing in order to disrupt communication more
efficiently, specially in a network of simple wireless devices such
as sensor nodes, which usually operate over a single frequency (or a
limited frequency band) and which cannot benefit from the use of
spread spectrum or other more advanced technologies. This article
proposes the medium access (MAC) protocol \textsc{AntiJam} that is
provably robust against a powerful reactive adversary who can jam a
$(1-\varepsilon)$-portion of the time steps, where $\varepsilon$ is
an arbitrary constant. The adversary uses carrier sensing to make
informed decisions on when it is most harmful to disrupt
communications; moreover, we allow the adversary to be adaptive and
to have complete knowledge of the entire protocol history. Our MAC
protocol is able to make efficient use of the non-jammed time
periods and achieves an asymptotically optimal,
$\Theta{(1)}$-competitive throughput in this harsh scenario. In
addition, \textsc{AntiJam} features a low convergence time and has
good fairness properties. Our simulation results validate our
theoretical results and also show that our algorithm manages to
guarantee constant throughput where the 802.11 MAC protocol
basically fails to deliver any packets.
\end{abstract}

%\section{Blackboard}
%
%\begin{enumerate}
%\item remove arxiv references
%\item include jins new simulations on gamma
%\item read through everything again, it was very messy!! (double lemmas, not everything introduced etc.!)
%\item make consistent: varepsilon or epsilon, textsc, etc.
%\end{enumerate}

\section{Introduction}\label{sec:introduction}

Disruptions of the communications over a shared medium---either
because of interference of concurrent transmissions or
intentionally---are a central challenge in wireless computing. It is
well-known that already simple jamming attacks---without any special
hardware---constitute a threat for the widely used IEEE~802.11 MAC
protocol. Due to the problem's relevance, there has been a
significant effort to cope with such disruption problems both from
the industry and the academia side and much progress has been made
over the last years on how to deal with different jammer types.

While simple oblivious jammers are well-understood today and many
countermeasures exist, this article goes an important step further
and studies MAC protocols against ``smart'' jammers. In particular,
we argue that adversaries may behave in an adaptive and reactive
manner: \emph{adaptive} in the sense that their decisions on whether
to jam at a certain moment in time can depend on the protocol
history; and \emph{reactive} in the sense that the adversary can
perform physical carrier sensing (which is part, e.g., of the 802.11
standard) to learn whether the channel is currently idle or not, and
jam the medium depending on these measurements.

This article presents the first medium access (MAC) protocol called
\textsc{AntiJam} that makes effective use of the few and arbitrarily
distributed non-jammed time periods, and achieves a provable
throughput despite the presence of such a strong reactive jammer. As
we will see, the throughput is asymptotically optimal, i.e., a
constant fraction of the non-jammed time period is used for
successful transmissions. Besides this interesting theoretic result,
our protocol is simple to implement and performs well also in the
average case. Also, worth to note is that our approach is at the MAC
level and may be used in conjunction with some of the anti-jamming
techniques developed at the physical layer (e.g., frequency hopping,
spread spectrum).

\subsection{Related Work}%\label{sec:introductio}

Researchers have studied the problem of unintentional and malicious
interference in wireless networks for several years now,
e.g.,~\cite{Chiang:mobicom07,citer1, Law:sasn05, Li:infocom07,
LiuNST07, NavdaBGR07, NegiP03, ThuenteA06,Xu:mobihoc05}. Classic
defense mechanisms operate on the physical layer
\cite{LiuNST07,NavdaBGR07} and there exist approaches both to {\em
avoid} as well as to {\em detect} jamming. Spread spectrum and
frequency hopping technologies have been shown to be very effective
to avoid jamming with widely spread signals. However, IEEE 802.11
variants spread signals with smaller factors~\cite{Brown:mobihoc06}
(IEEE 802.11b uses a narrow spreading factor of 11~\cite{IEEE99}).
As jamming strategies can come in many different flavors, detecting
jamming activities by simple methods based on signal strength,
carrier sensing, or packet delivery ratios has turned out to be
quite difficult \cite{Li:infocom07}.

Recent work has also studied \emph{MAC layer strategies} against
jamming, including coding strategies (e.g.,
\cite{Chiang:mobicom07}), channel surfing and spatial retreat (e.g.,
\cite{Alnifie:Q2SWinet07,Xu:wws04}), or mechanisms to hide messages
from a jammer, evade its search, and reduce the impact of corrupted
messages (e.g., \cite{Wood:secon07}). Unfortunately, these methods
do not help against an adaptive jammer with {\em full} information
about the history of the protocol, like the one considered in our
work.

In the theory community, work on MAC protocols has mostly focused on
efficiency. Many of these protocols are \emph{random backoff
protocols} (e.g.,
\cite{Bender05,Chlebus06,CR06,Hastad96,Raghavan99}) that do not take
jamming activity into account and, in fact, are not robust against
it (see \cite{singlehop08} for more details). Also some theoretical
work on \emph{jamming} is known (e.g.,~\cite{citer2} for a short
overview). There are two basic approaches in the literature. The
first assumes randomly corrupted messages (e.g.~\cite{PP05}), which
is much easier to handle than adaptive adversarial jamming
\cite{Raj08}. The second line of work either bounds the number of
messages that the adversary can transmit or disrupt with a limited
energy budget (e.g.~\cite{GGN06,KBKV06}), or bounds the number of
channels the adversary can jam (e.g.~\cite{seth09,dcoss09}).

The protocols in, e.g.,~\cite{KBKV06} can tackle adversarial jamming
at both the MAC and network layers, where the adversary may not only
be jamming the channel but also introducing malicious (fake)
messages (possibly with address spoofing). However, they depend on
the fact that the adversarial jamming budget is finite, so it is not
clear whether the protocols would work under heavy continuous
jamming. (The result in \cite{GGN06} seems to imply that a jamming
rate of 0.5 is the limit whereas the handshaking mechanisms in
\cite{KBKV06} seem to require an even lower jamming rate.)

Our work is motivated by the results in \cite{Raj08} and
\cite{singlehop08}. In \cite{Raj08} it is shown that an adaptive
jammer can dramatically reduce the throughput of the standard MAC
protocol used in IEEE 802.11 with only limited energy cost on the
adversary side. Awerbuch et al. \cite{singlehop08} initiated the
study of throughput-competitive MAC protocols under continuously
running, adaptive jammers, and presented a protocol that achieves a
high performance under adaptive jamming.

In this article, we extend the model and result from
\cite{singlehop08} in a crucial way: we allow the jammer to be
\emph{reactive}, i.e., to listen to the current channel state in
order to make smarter jamming decisions. We believe that the
reactive jammer model is much more realistic and hence that our
study is of practical importance. For example, by sensing the
channel, the adversary may avoid wasting energy by not jamming idle
rounds. Note however that depending on the protocol, it may still
make sense for the adversary to jam idle rounds, e.g., to influence
the protocol execution. Indeed, due to the large number of possible
strategies a jammer can pursue, the problem becomes significantly
more challenging than the non-reactive version: Not only is the
analysis more involved, but also key modifications to the protocol
in~\cite{singlehop08} were needed. While we still build upon the
algorithmic ideas presented in~\cite{singlehop08}, in order to avoid
asymmetries, our \textsc{AntiJam} protocol seeks to synchronize the
nodes' sending probabilities. (This has the desirable side effect of
an improved fairness.) While our formal analysis confirms our
expectations that the overall throughput under reactive jammers is
lower than the throughput obtainable against non-reactive jammers,
we are still able to prove a good, constant-competitive performance,
which is also confirmed by our simulation study. As a final remark,
although this article focuses on single-hop environments, our first
insights indicate that \textsc{AntiJam}-like strategies can also be
used in multi-hop settings (see also the recent extension
of~\cite{singlehop08} to unit disk graphs~\cite{multihop10}).

\subsection{Model}\label{sec:model}

We study a wireless network that consists of $n$ honest and reliable
simple wireless devices (e.g., sensor nodes) that are within the
transmission range of each other and which communicate over a single
frequency (or a limited, narrow frequency band). We assume a
back-logged scenario where the nodes continuously contend for
sending a packet on the wireless channel. A node may either transmit
a message or sense the channel at a time step, but it cannot do
both, and there is no immediate feedback mechanism telling a node
whether its transmission was successful. A node sensing the channel
may either (i) sense an \emph{idle channel} (in case no other node
is transmitting at that time), (ii) sense a \emph{busy channel} (in
case two or more nodes transmit at the time step), or (iii)
\emph{receive} a packet (in case exactly one node transmits at the
time step).

In addition to these nodes there is an adversary. We allow the
adversary to know the protocol and its entire history and to use
this knowledge in order to jam the wireless channel at will at any
time (i.e, the adversary is \emph{adaptive}). Whenever it jams the
channel, all nodes will notice a busy channel. However, the nodes
cannot distinguish between the adversarial jamming or a collision of
two or more messages that are sent at the same time. We assume that
the adversary is only allowed to jam a $(1-\varepsilon)$-fraction of
the time steps, for an arbitrary constant $0< \varepsilon \leq 1$.

Moreover, we allow the jammer to be \emph{reactive}:  it is allowed
to make a jamming decision \emph{after} it knows the actions of the
nodes at the current step. In other words, reactive jammers can
determine (through physical carrier sensing) whether the channel is
currently idle or busy (either because of a successful transmission,
a collision of transmissions, or too much background noise) and can
instantly make a jamming decision based on that information. Those
jammers arise in scenarios where, for example, encryption is being
used for communication and where the jammer cannot distinguish
between an encrypted package and noise in the channel.

In addition, we allow the adversary to perform \emph{bursty}
jamming. More formally, an adversary is called
$(T,1-\varepsilon)$-{\em bounded} for some $T\in\mathbb{N}$ and $0 <
\varepsilon < 1$ if for any time window of size $w \geq T$ the
adversary can jam at most $(1-\varepsilon) w$ of the time steps in
that window.

The network scenario described above arises, for example, in sensor
networks, which consist of simple wireless nodes usually running on
a single frequency and which cannot benefit from more advanced
anti-jamming techniques such as frequency hopping or spread
spectrum. In such scenarios, a jammer will also most probably run on
power-constrained devices (e.g., solar-powered batteries), and hence
will not have enough power to continuously jam over time (note that
the time window threshold $T$ can be chosen large enough to
accommodate the respective jamming pattern).

This article studies \emph{competitive} MAC protocols. A MAC
protocol is called $c$-competitive against some
$(T,1-\varepsilon)$-bounded adversary (with high probability or on
expectation) if, for any sufficiently large number of time steps,
the nodes manage to perform successful message transmissions in at
least a $c$-fraction of the time steps not jammed by the adversary
(with high probability or on expectation).

Our goal is to design a {\em symmetric local-control} MAC protocol
(i.e., there is no central authority controlling the nodes, and the
nodes have symmetric roles at any point in time) that is
$O(1)$-competitive against any $(T,1-\varepsilon)$-bounded
adversary. The nodes do not know $\varepsilon$, but we do allow them
to have a very rough upper bound of the number $n$ and $T$. More
specifically, we will assume that the nodes have a common parameter
$\gamma=O(1/(\log T + \log\log n))$. This is still scalable, since
such an estimate leaves room for a super-polynomial change in $n$
and a polynomial change in $T$ over time, so it does not make the
problem trivial (as would be the case if the nodes knew constant
factor approximations of $n$ or $T$).

\subsection{Our Contributions}

This article introduces and analyzes the medium access protocol
\textsc{AntiJam}. \textsc{AntiJam} is robust to a strong adaptive
and reactive jammer who can block a constant fraction of the time
and thus models a large range of (intentional and unintentional)
interference models. Nevertheless, we can show that the
\textsc{AntiJam} MAC protocol achieves a high throughput performance
by exploiting any non-blocked time intervals effectively. The main
theoretical contribution is the derivation of the following theorem
that shows that \textsc{AntiJam} is asymptotically optimal in the
sense that a constant fraction of the non-jammed execution time is
used for successful transmissions:

\begin{theorem}~\label{th:main}
The MAC protocol is $\Theta(1)$-competitive w.h.p.\footnote{With
high probability, i.e., with probability at least $1-1/n^c$, where
$c$ is a constant. As $n$ grows to infinity, the probability tends
to 1.} under any $(T,1-\varepsilon)$-bounded adversary for some
constant $\varepsilon$ if the protocol is executed for at least
$\Theta(\frac{1}{\varepsilon} \log N \max\{T, \frac{1}{\varepsilon
\gamma^2} \log^3 N\})$ many time steps.
\end{theorem}

We believe that \textsc{AntiJam} is interesting also from a
practical point of view, as the basic protocol is very simple. We
also report on our simulation results. It turns out that
\textsc{AntiJam} is able to benefit from the rare and
hard-to-predict time intervals where the shared medium is available.
Moreover, \textsc{AntiJam} converges fast and allocates the shared
medium \emph{fairly} to the nodes.

\subsection{Article Organization}

The remainder of this article is organized as follows.
Section~\ref{sec:algo} introduces the \textsc{AntiJam} MAC protocol.
Subsequently, we present a formal analysis of the throughput
performance under reactive jamming (Section~\ref{sec:analysis}).
Section~\ref{sec:experiments} reports on the insights gained from
our simulation experiments. The article is concluded in
Section~\ref{sec:conclusion}.

\section{The \textsc{AntiJam} MAC Protocol}\label{sec:algo}

The basic ideas of the \textsc{AntiJam} MAC protocol are inspired by
the protocols described in~\cite{idlesense} (which also uses access
probabilities depending on the ratio between idling and successful
time slots) and particularly~\cite{singlehop08}. However, the
algorithm in~\cite{singlehop08} does not achieve a good performance
under reactive jammers, which is due to the asymmetric access
probabilities. Therefore, in our protocol, we explicitly try to
equalize access probabilities, which also improves fairness among
the nodes.

Each node $v$ maintains a time window threshold estimate $T_v$ and a
counter $c_v$. The parameter $\gamma$ is the same for every node and
is set to some sufficiently small value in $O(1/(\log T$ $+ \log
\log n))$. Thus, we assume that the nodes have some polynomial
estimate of $T$ and even rougher estimate of $n$. Let $\hat{p}$ be
any constant so that $0<\hat{p} \le 1/24$. Initially, every node $v$
sets $T_v:=1$, $c_v:=1$ and $p_v:=\hat{p}$. Afterwards, the protocol
works in synchronized time steps. We assume synchronized time steps
for the analysis, but a non-synchronized execution of the protocol
would also work as long as all nodes operate at roughly the same
speed.

The basic protocol idea is simple. Suppose that each node $v$
decides to send a message at the current time step with probability
$p_v$ with $p_v \le \hat{p}$. Let $p=\sum_v p_v$, $q_0$ be the
probability that the channel is idle and $q_1$ be the probability
that exactly one node is sending a message. The following claim
appeared originally in \cite{singlehop08}. It states that if $q_0 =
\Theta(q_1)$, then the cumulative sending probability $p$ is
constant, which in turn implies that at any non-jammed time step we
have constant probability of having a successful transmission. Hence
our protocol aims at adjusting the sending probabilities $p_v$ of
the nodes such that $q_0 = \Theta(q_1)$, in spite of the reactive
adversarial jamming activity. This will be achieved by using a
multiplicative increase/decrease game for the probabilities $p_v$
and by synchronizing all the nodes, both in terms of sending
probabilities and their own estimates on the time window threshold
estimate $T_v$'s, at every successful transmission.
\begin{claim} \label{cl_relation}
$q_0 \cdot p \le q_1 \le \frac{q_0}{1-\hat{p}} \cdot p$.
\end{claim}

Now we present our  \textsc{AntiJam} protocol:
\begin{shaded}
In each step, each node $v$ does the following. $v$ decides with
probability $p_v$ to send a message along with a tuple: $(p_v,c_v,
T_v)$. If it decides not to send a message, it checks the following
two conditions:
\begin{enumerate}
\item If $v$ senses an idle channel, then $p_v:=
\min\{(1+\gamma)p_v, \hat{p}\}$ and $T_v:=T_v-1$.

\item If $v$ successfully receives a message along with the tuple
of $(p_{new},c_{new}, T_{new})$, then $p_v:=(1+\gamma)^{-1}p_{new}$,
$c_v:=c_{new}$, and $T_v:=T_{new}$.
\end{enumerate}
Afterwards, $v$ sets $c_v:=c_v+1$. If $c_v>T_v$ then it does the
following: $v$ sets $c_v:=1$, and if there was no idle step among
the past $T_v$ time steps, then $p_v:= (1+\gamma)^{-1} p_v$ and
$T_v:=T_v+2$.
\end{shaded}

\section{Analysis}\label{sec:analysis}

Now we restate Theorem~\ref{th:main} more precisely. We will prove
this more technical version of Theorem~\ref{th:main}. Let
$N=\max\{T,n\}$.
\begin{theorem} \label{th_main}
The \textsc{AntiJam} protocol is
$e^{-\Theta(1/\varepsilon^2)}$-competitive w.h.p.~under any
$(T,1-\varepsilon)$-bounded adversary if the protocol is executed
for at least $\Theta(\frac{1}{\varepsilon} \log N \max\{T,
(e^{\delta/\varepsilon^2}/\varepsilon \gamma^2) \log^3 N\})$ many
time steps, where $\delta$ is a sufficiently large constant.
\end{theorem}

In our analysis, we will make use of the following well-known
relations.
\begin{lemma} \label{lem:euler}
For all $0<x<1$ it holds that
\[
  e^{-x/(1-x)} \le 1-x \le e^{-x}
\]
\end{lemma}
\begin{lemma} \label{lem_chernoff}
Consider any set of binary random variables $X_1,$ $\ldots,X_n$.
Suppose that there are values $p_1,\ldots,p_n \in [0,1]$ with $\E [
\prod_{i \in S} X_i] \le \prod_{i \in S} p_i$ for every set $S
\subseteq \{1,\ldots,n\}$. Then it holds for $X=\sum_{i=1}^n X_i$
and $\mu = \sum_{i=1}^n p_i$ and any $\delta>0$ that
\[
  \Pr[X \ge (1+\delta)\mu] \le \left( \frac{e^{\delta}}{(1+\delta)^{1+\delta}}
  \right)^{\mu} \le e^{-\frac{\delta^2 \mu}{2(1+\delta/3)}}.
\]
If, on the other hand, it holds that $\E [ \prod_{i \in S} X_i] \ge
\prod_{i \in S} p_i$ for every set $S \subseteq \{1,\ldots,n\}$,
then it holds for any $0 < \delta < 1$ that
\[
  \Pr[X \le (1-\delta)\mu] \le \left( \frac{e^{-\delta}}{(1-\delta)^{1-\delta}}
  \right)^{\mu} \le e^{-\delta^2 \mu / 2}.
\]
\end{lemma}

Let $V$ be the set of all nodes. Let $p_t(v)$ be node $v$'s access
probability $p_v$ at the beginning of the $t$-th time step.
Furthermore, let $p_t = \sum_{v \in V} p_t(v)$. Let $I$ be a time
frame consisting of $\frac{\alpha}{\varepsilon} \log N$ {\em
subframes} $I'$ of size $f=\max\{T, \frac{\alpha
\beta^2}{\varepsilon \gamma^2} e^{\delta/\varepsilon^2} \log^3 N\}$,
where $\alpha$, $\beta$ and $\delta$ are sufficiently large
constants. Let $F=\frac{\alpha}{\varepsilon} \log N \cdot f$ denote
the size of $I$.

First, we will derive some simple facts on the behavior of
\textsc{AntiJam}. We then show that given a certain minimal initial
cumulative probability $p_t$ in a subframe, the cumulative
probability cannot be smaller at the end of the subframe. We proceed
to show that \textsc{AntiJam} performs well in time periods in which
$p_t$ is bounded by $\delta/\varepsilon^2$ for some constant
$\delta$. Finally, we show that for any jamming strategy,
\textsc{AntiJam} has a cumulative probability of $p_t\leq
\delta/\varepsilon^2$ for most of the time, which yields our main
theorem.

We start with some simple facts. Fact~\ref{fa:access1} shows that
the protocol synchronizes the sending probabilities of the nodes (up
to a factor of $(1+\gamma)$), and that all values $c_v$ and $T_v$
are also synchronized.
\begin{fact} \label{fa:access1}
Right after a successful transmission of the tuple $(p',c',T')$,
$(p_v,c_v,T_v)=((1+\gamma)^{-1} p',c',T')$ for all receiving nodes
$v$ and $(p_u,c_u,T_u)=(p',c',T')$ for the sending node $u$. In
particular, for any time step $t$ after a successful transmission by
node $u$, $(c_v,T_v)=(c_w,T_w)$ for all nodes $v,w \in V$.
\end{fact}

The next fact follows from the protocol and Fact~\ref{fa:access1},
and they help one understand how the cumulative probabilities vary
over time with successful transmissions, idle time steps, etc.

\begin{fact} \label{fa:access3}
For any time step $t$ after a successful transmission or a
well-initialized state of the protocol (in which
$(p_v,c_v,T_v)=(\hat{p},1,1)$ for all nodes $v$) it holds:

\noindent {\bf\em 1.} If the channel is {\em idle} at time $t$ then
$(i)$ if $p_v=\hat{p}$ for all $v$, then $p_{t+1} = p_t$;
%\item\label{fact:idle2} If the channel is idle at time $t$ and
 $(ii)$ if $p_u=\hat{p}$ and $p_v = (1+\gamma)^{-1} \hat{p}$ for all nodes
$v\not=u$, then $p_{t+1} = (1+\gamma-O(1/n))p_t$ (because all nodes
except for $u$ increase their sending probability by a factor
$(1+\gamma)$ from $\hat{p}/(1+\gamma)$.); or
%\item\label{fact:idle3} If the channel is idle at time $t$ and
$(iii)$ if $p_v < \hat{p}$ for all nodes $v$, then $p_{t+1} =
(1+\gamma) p_t$.

\noindent {\bf\em 2.} If there is a {\em successful transmission} at
time $t$, and if $c_v\leq T_v$ or there was an idle time step in the
previous $T_v$ rounds, then $(i)$ if the sender is the same as the
last successful sender, then $p_{t+1} = p_t$ (because for the sender
$u$, $p_u(t+1)=p_u(t)$, and the other nodes remain at
$p_u(t+1)/(1+\gamma)=p_u(t)/(1+\gamma)$.);
if $(ii)$ the sender $w$
is different from the last successful sender $u$ and $p_v=\hat{p}$
for all nodes $v$ (including $u$ and $w$), then $p_{t+1} =
(1+\gamma-O(1/n))^{-1} p_t$
%when ignoring the case that $c_v>T_v$
(all nodes except $w$ reduce their sending probability.); or
%
%If there is a successful transmission at time $t$ and
$(iii)$ if the sender $w$ is different from the last successful
sender $u$ and $p_v<\hat{p}$ for at least one node $v$ (including
$u$ and $w$), then $p_{t+1} = (1+\gamma)^{-1} p_t$
%when ignoring the case that $c_v>T_v$
(because at time $t$, for all nodes $v\neq u$:
$p_v(t)=p_{u}(t)/(1+\gamma)$; subsequently, $p_{w}(t+1)=p_{w}(t)$
and for all nodes $v\neq w$: $p_v(t+1)=p_{w}(t+1)/(1+\gamma)$.)

%\item\label{fact:busy}
\noindent {\bf\em 3.} If the channel is {\em busy} at time $t$, then
$p_{t+1} = p_t$ when ignoring the case that $c_v>T_v$.

%\item\label{fact:noidle}
Whenever $c_v>T_v$ and there has not been an idle time step during
the past $T_v$ steps, then $p_{t+1}$ is, in addition to the actions
specified in the two cases above, reduced by a factor of
$(1+\gamma)$.
%\end{enumerate}
\end{fact}

We can now prove the following crucial lemma.
\begin{lemma} \label{lem_prob}
For any subframe $I'$ in which initially $p_{t_0} \ge 1/(f^2
(1+\gamma)^{\sqrt{2f}})$, the last time step $t$ of $I'$ again
satisfies $p_t \ge 1/(f^2 (1+\gamma)^{\sqrt{2f}})$, w.h.p.
\end{lemma}
\begin{proof}
We start with the following claim about the maximum number of times
the nodes decrease their probabilities in $I'$ due to $c_v>T_v$.
\begin{claim} \label{cl_Tinc}
If in subframe $I'$ the number of idle time steps is at most $k$,
then every node $v$ increases $T_v$ by 2 at most $k/2+\sqrt{f}$ many
times.
\end{claim}
\begin{proof}
Only idle time steps reduce $T_v$. If there is no idle time step
during the last $T_v$ many steps, $T_v$ is increased by 2. Suppose
that $k=0$. Then the number of times a node $v$ increases $T_v$ by 2
is upper bounded by the largest possible $\ell$ so that
$\sum_{i=0}^{\ell} T_v^0+2i \le f$, where $T_v^0$ is the initial
size of $T_v$. For any $T_v^0 \ge 1$, $\ell \le \sqrt{f}$, so the
claim is true for $k=0$. At best, each additional idle time step
allows us to reduce all thresholds for $v$ by 1, so we are searching
for the maximum $\ell$ so that $\sum_{i=0}^{\ell}
\max\{T_v^0+2i-k,1\} \le f$. This $\ell$ is upper bounded by $k/2 +
\sqrt{f}$, which proves our claim.
\end{proof}

This claim allows us to show the following claim.
\begin{claim} \label{cl_pup}
Suppose that for the first time step $t_0$ in $I'$, $p_{t_0} \in
[1/(f^2 (1+\gamma)^{\sqrt{2f}}), 1/f^2]$. Then there is a time step
$t$ in $I'$ with $p_t \ge 1/f^2$, w.h.p.
\end{claim}
\begin{proof}
Suppose that there are $g$ non-jammed time steps in $I'$. Let $k_0$
be the number of these steps with an idle channel and $k_1$ be the
number of these steps with a successful message transmission.
Furthermore, let $k_2$ be the maximum number of times a node $v$
increases $T_v$ by 2 in $I'$. If all time steps $t$ in $I'$ satisfy
$p_t < 1/f^2$, then it must hold that
\[
  k_0 - \log_{1+\gamma} (1/p_{t_0}) \le k_1+k_2.
\]
This is because no $v$ has reached a point with $p_t(v)=\hat{p}$ in
this case, so Fact~\ref{fa:access3} %(Number~\ref{fact:idle3})
implies that for each time step $t$ with an idle channel, $p_{t+1} =
(1+\gamma) p_t$. Thus, at most $\log_{1+\gamma} (1/p_{t_0})$ time
steps with an idle channel would be needed to get $p_t$ to $1/f^2$,
and then there would have to be a balance between further increases
(that are guaranteed to be caused by an idle channel) and decreases
(that might be caused by a successful transmission or the case
$c_v>T_v$) of $p_t$ in order to avoid the case $p_t \ge 1/f^2$. The
number of times we can allow an idle channel is maximized if all
successful transmissions and cases where $c_v>T_v$ cause a reduction
of $p_t$. So we need $k_0 - \log_{1+\gamma} (1/p_{t_0}) \le k_1+k_2$
to hold to avoid the case $p_t \ge 1/f^2$ somewhere in $I'$.

We know from Claim~\ref{cl_Tinc} that $k_2 \le k_0/2 + \sqrt{f}$.
Hence,
\begin{footnotesize}
\begin{eqnarray*}
  k_0 & \le & 2 \log_{1+\gamma} f + \sqrt{f} + k_1 + k_0/2 + \sqrt{f} \\
\Rightarrow \quad k_0 & \le & 4 \log_{1+\gamma} f + 2k_1 + 4\sqrt{f}
\end{eqnarray*}
\end{footnotesize}
Suppose that $4 \log_{1+\gamma} f + 4\sqrt{f} \le
\varepsilon f/4$, which is true if $f=\Omega(1/\varepsilon^2)$ is
sufficiently large (which is true for $\varepsilon=\Omega(1/\log^3
N)$). Since $g \ge \varepsilon f$ due to our adversarial model, it
follows that we must satisfy $k_0 \le 2k_1 + g/4$.

Certainly, for any time step $t$ with $p_t \le 1/f^2$,
\begin{eqnarray*}
  \Pr[\ge 1 \mbox{ message transmitted at $t$}] & \le & 1/f^2
\end{eqnarray*}
Suppose for the moment that no time step is jammed in $I'$. Then
$\E[k_1] \le (1/f^2)f = 1/f$. In order to prove a bound on $k_1$
that holds w.h.p., we can use the general Chernoff bounds stated
above. For any step $t$, let the binary random variable $X_t$ be 1
if and only if at least one message is sent at time $t$ and $p_t \le
1/f^2$. Then
\begin{footnotesize}
\begin{eqnarray*}
  \Pr[X_t=1] & = & \Pr[p_t \le 1/f^2] \cdot \Pr[\ge 1 \mbox{ msg sent} \mid p_t \le
  1/f^2] \\
  & \le & 1/f^2
\end{eqnarray*}
\end{footnotesize} and it particularly holds that for any set $S$ of
time steps prior to some time step $t$ that
\[
  \Pr[X_t = 1 \mid \prod_{s \in S} X_s=1] \le 1/f^2
\]
Then, we have
\begin{footnotesize}
\begin{eqnarray*}
     \Pr[\prod_{s \in S} X_s=1] & = & \Pr[X_1 = 1]\cdot \Pr[X_2 = 1|X_1 = 1] \\
    &\cdot& \Pr[X_3 = 1| \prod_{s = 1,2} X_s = 1]  \\
    &{\cdot...\cdot}\\
    &\cdot& \Pr[X_{|S|} = 1| \prod_{s = 1,2,...,|S|-1} X_s = 1] \\
     & \le & (1/f^2)^{|S|}
\end{eqnarray*}
\end{footnotesize}
and
\[
 \E[\prod_{s \in S} X_s=1] =  \Pr[\prod_{s \in S} X_s = 1] \le  (1/f^2)^{|S|}
\]

Thus, the Chernoff bounds and our choice of $f$ imply that either
$\sum_{t \in I'} X_t < \varepsilon f/4$ and $p_t \le 1/f^2$
throughout $I'$ w.h.p., or there must be a time step $t$ in $I'$
with $p_t > 1/f^2$ which would finish the proof. Therefore, unless
$p_t > 1/f^2$ at some point in $I'$, $k_1<\varepsilon f/4$ and $k_0
> (1-\varepsilon/4)f$ w.h.p. As the reactive adversary can now reduce
$k_0$ by at most $f-g$ when leaving $g$ non-jammed steps, it follows
that for any adversary, $k_0 > (1-\varepsilon/4)f - (f-g) = g-
(\varepsilon/4)f$. That, however, would violate our condition above
that $k_0 \le 2k_1 + g/4$ as that can only hold given the bounds on
$g$ and $k_1$ if $k_0 \le g-(\varepsilon/4)f$.

Note that the choice of $g$ is not oblivious as the adversary may
{\em adaptively} decide to set $g$ based on the history of events.
Hence, we need to sum up the probabilities over all adversarial
strategies of selecting $g$ in order to show that none of them
succeeds, but since there are only $f$ many, and for each the
claimed property holds w.h.p., the claim follows.
\end{proof}

Similar to this claim, we can also prove the following claim.

\begin{claim} \label{cl_plp}
Suppose that for the first time step $t_0$ in $I'$, $p_{t_0} \ge
1/f^2$. Then there is no time step $t$ in $I'$ with $p_t <
\frac{1}{f^2 (1+\gamma)^{\sqrt{2f}}}$, w.h.p.
\end{claim}
\begin{proof}
Consider some fixed time step $t$ in $I'$ and let $I''=(t_0,t]$.
Suppose that there are $g$ non-jammed time steps in $I''$. If $g \le
\beta \log N$ for a (sufficiently large) constant $\beta$, then it
follows for the probability $p_t$ at the end of $I''$ due to
Claim~\ref{cl_Tinc} that
\[
  p_t \ge \frac{1}{f^2} \cdot (1+\gamma)^{-(2\beta \log N + \sqrt{f})}
    \ge \frac{1}{f^2 (1+\gamma)^{\sqrt{2f}}}
\]
given that $\varepsilon = \Omega(1/\log^3 N)$, because at most
$\beta \log N$ decreases of $p_t$ can happen due to a successful
transmission and at most $\beta \log N/2 + \sqrt{f}$ further
decreases of $p_t$ can happen due to exceeding $T_v$.

So suppose that $g > \beta \log N$. Let $k_0$ be the number of these
steps with an idle channel and $k_1$ be the number of these steps
with a successful message transmission. Furthermore, let $k_2$ be
the maximum number of times a node $v$ increases $T_v$ in $I''$. If
$p_t < \frac{1}{f^2 (1+\gamma)^{\sqrt{2f}}}$ then it must hold that
\[
  k_0 \le k_1 + k_2
\]
Since $k_2 \le k_0/2 + \sqrt{f}$, this implies that $k_0 \le 2k_1 +
2\sqrt{f} \le 2k_1 + g/4$. Thus, we are back to the case in the
proof of Claim~\ref{cl_pup}, which shows that $k_0 \le 2k_1 + g/4$
does not hold w.h.p., given that $g > \beta \log N$ and we never
have the case in $I''$ that $p_t
> 1/f^2$.

If there is a step $t'$ in $I''$ with $p_{t'} > 1/f^2$, we prune
$I''$ to the interval $(t',t]$ and repeat the case distinction
above. As there are at most $f$ time steps in $I''$, the claim
follows.
\end{proof}

Combining Claims~\ref{cl_pup} and \ref{cl_plp} completes the proof
of Lemma~\ref{lem_prob}.
\end{proof}

Lemma~\ref{lem_enough_useful} shows that for times of low cumulative
probabilities, \textsc{AntiJam} yields a good performance.
\begin{lemma}\label{lem_enough_useful}
Consider any subframe $I'$, and let $\delta > 1$ be a sufficiently
large constant. Suppose that at the beginning of $I'$, $p_{t_0} \ge
1/(f^2 (1+\gamma)^{\sqrt{2f}})$ and $T_v \le \sqrt{F}/2$ for every
node $v$. If $p_t \le \delta/\varepsilon^2$ for at least half of the
non-jammed time steps in $I'$, then \textsc{AntiJam} is at least
$\frac{\delta}{8(1-\hat{p})\varepsilon^2}
e^{-\delta/(1-\hat{p})\varepsilon^2}$-competitive in $I'$.
\end{lemma}
\begin{proof}
A time step $t$ in $I$ is called {\em useful} if we either have an
idle channel or a successful transmission at time $t$ (i.e., the
time step is not jammed and there are no collisions) and $p_t \le
\delta/\varepsilon^2$. Let $k$ be the number of useful time steps in
$I'$. Furthermore, let $k_0$ be the number of useful time steps in
$I'$ with an idle channel, $k_1$ be the number of useful time steps
in $I'$ with a successful transmission and $k_2$ be the maximum
number of times a node $v$ reduces $p_v$ in $I'$ because of
$c_v>T_v$. Recall that $k=k_0+k_1$. Moreover, the following claim
holds:

\begin{claim} \label{cl:useful1}
If $n \ge (1+\gamma) \delta/(\varepsilon^2 \hat{p})$, then
\[
  k_0 - \log_{1+\gamma} (\delta/(\varepsilon^2 \cdot p_{t_0})) \le k'_1+k_2
\]
where $k'_1$ is the number of useful time steps with a successful
transmission in which the sender is different from the previously
successful sender.
\end{claim}
\begin{proof}
According to Fact~\ref{fa:access3}, $p_v \in [(1+\gamma)^{-1} p, p]$
for some access probability $p$ for all time steps in $I'$. Hence,
if $p_t \le \delta/\varepsilon^2$ and $n \ge
(1+\gamma)\delta/(\varepsilon^2 \hat{p})$, then $p_v(t) \leq
\hat{p}/(1+\gamma)$. This implies that whenever there is a useful
time step $t \in I$ with an idle channel, then $p_{t+1} = (1+\gamma)
p_t$. Thus, it takes at most $\log_{1+\gamma} (\delta/(\varepsilon^2
\cdot p_{t_0}))$ many useful time steps with an idle channel to get
from $p_{t_0}$ to a cumulative probability of at least
$\delta/\varepsilon^2$. On the other hand, each of the $k'_1$
successful transmissions reduces the cumulative probability by
$(1+\gamma)$. Therefore, once the cumulative probability is at
$\delta/\varepsilon^2$, we must have $k_0 \le k'_1+k_2$ since
otherwise there must be at least one useful time step where the
cumulative probability is more than $\delta/\varepsilon^2$, which
contradicts the definition of a useful time step.
\end{proof}

Since $p_{t_0} \ge 1/(f^2 (1+\gamma)^{\sqrt{2f}})$ it holds that
\[
  \log_{1+\gamma} (\delta/(\varepsilon^2 \cdot p_{t_0}))
  \le \log_{1+\gamma} (\delta f^2/\varepsilon^2) + \sqrt{2f}
\]
From Lemma~\ref{cl_Tinc} we also know that $k_2 \le k_0/2 +
\sqrt{f}$. Hence,
\begin{footnotesize}
\begin{eqnarray*}
  k_0 & \le & 2k'_1 + 2\cdot\log_{1+\gamma} (\delta f^2/\varepsilon^2) + 2 \cdot (\sqrt{f} +
    \sqrt{2f})\\
  & \le & 2k'_1 + 6\sqrt{f}
\end{eqnarray*}
\end{footnotesize} if $f$ is sufficiently large. Also, $k_0 = k-k_1$
and $k'_1 \le k_1$. Therefore, $k-k_1 \le 2k_1 + 6\sqrt{f}$ or
equivalently,
\[
  k_1 \ge k/3 - 2\sqrt{f}
\]
It remains to find a lower bound for $k$.

\begin{claim} \label{cl:useful2}
Let $g$ be the number of non-jammed time steps $t$ in $I'$ with $p_t
\le \delta/\varepsilon^2$. If $g \ge \varepsilon f/2$ then
\[
  k \ge \frac{\delta}{2(1-\hat{p})\varepsilon^2}
  e^{-\delta/(1-\hat{p})\varepsilon^2} \cdot g
\]
w.h.p.
\end{claim}
\begin{proof}
Consider any $(T,1-\varepsilon)$-bounded jammer for $I'$. Suppose
that of the non-jammed time steps $t$ with $p_t \le
\delta/\varepsilon^2$, $s_0$ have an idle channel and $s_1$ have a
busy channel. It holds that $s_0+s_1 = g \ge \varepsilon f/2$. For
any one of the non-jammed time steps with an idle channel, the
probability that it is useful is one, and for any one of the
non-jammed time steps with a busy channel, the probability that it
is useful (in this case, that it has a successful transmission) is
at least
\begin{footnotesize}
\begin{eqnarray*}
  \sum_v p_v \prod_{w\not=v} (1-p_w) & \ge & \frac{1}{1-\hat{p}} \sum_v p_v
    \prod_w (1-p_w) \\
  & \ge & \frac{1}{1-\hat{p}} \sum_v p_v \prod_w e^{-p_w/(1-\hat{p})} \\
  & = & \frac{1}{1-\hat{p}} \sum_v p_v e^{-p/(1-\hat{p})} \\
  & = & \frac{p}{1-\hat{p}} e^{-p/(1-\hat{p})}
\end{eqnarray*}
\end{footnotesize} where $p$ is the cumulative probability at the
step. Since $p_t \le \delta/\varepsilon^2$, it follows that the
probability of a busy time step to be useful is at least
\[
  \frac{\delta}{(1-\hat{p})\varepsilon^2} e^{-\delta/(1-\hat{p})\varepsilon^2}
\]
Thus,
\[
  \E[k] \ge s_0 + \frac{\delta}{(1-\hat{p})\varepsilon^2}
  e^{-\delta/(1-\hat{p})\varepsilon^2} s_1
  \]
  \[
  \ge \frac{\delta}{(1-\hat{p})\varepsilon^2} e^{-\delta/(1-\hat{p})\varepsilon^2} \cdot g
\]
since $k$ is minimized for $s_0=0$ and $s_1 = g$.

Since our lower bound for the probability of a busy step to be
useful holds independently for all non-jammed busy steps $t$ with
$p_t \le \delta/\varepsilon^2$ and $E[k] \ge \alpha \log N$ for our
choice of $g$, it follows from the Chernoff bounds that $k \ge
\E[k]/2$ w.h.p.
\end{proof}

From Claim~\ref{cl:useful2} it follows that
\[
  k_1 \ge (\frac{\delta}{2(1-\hat{p})\varepsilon^2}
  e^{-\delta/(1-\hat{p})\varepsilon^2} \cdot g)/3 - 2\sqrt{f}
\]
w.h.p., which completes the proof of Lemma~\ref{lem_enough_useful}.
\end{proof}

It remains to consider the case that for less than half of the
non-jammed time steps $t$ in $I'$, $p_t \le \delta/\varepsilon^2$.
Fortunately, this does not happen w.h.p.

\begin{lemma} \label{lem:high}
Suppose that at the beginning of $I'$, $T_v \le \sqrt{F}/2$ for
every node $v$. Then at most half of the non-jammed time steps $t$
can have the property that $p_t > \delta/\varepsilon^2$ w.h.p.
\end{lemma}
\begin{proof}
Recall from Fact~\ref{fa:access3} that as long as the access
probabilities of the nodes do not hit $\hat{p}$, the cumulative
probability only changes by a $(1+\gamma)$-factor in both
directions. Suppose that $\delta$ is selected so that
$\delta/\varepsilon^2$ represents one of these values. Let $H$ be
the set of time steps $t \in I'$ with the property that either
$p_t=\delta/\varepsilon^2$ and the channel is idle or $p_t \ge
(1+\gamma)\delta/\varepsilon^2$. Now, we define a step $t$ to be
useful if $t \in H$ and there is either an idle channel or a
successful transmission at $t$. Let $k$ be the number of useful time
steps in $H$. Furthermore, let $k_0$ be the number of useful time
steps with an idle channel, $k_1$ be the number of useful time steps
with a successful transmission and $k_2$ be the maximum number of
times a node $v$ reduces $p_v$ in $H$ because of $c_v>T_v$. It holds
that $k=k_0+k_1$.

Let us cut the time steps in $H$ into {\em passes} where each pass
$(t,p,S)$ consists of a time step $t$ with $p_t=p$ in which there is
an idle channel (or $t$ is the beginning of $I'$ if there is no such
idle channel in $I'$) and $S$ is the sequence of all non-idle time
steps $t'>t$ with $p_{t'} = (1+\gamma)p$ following $t$ until a time
step $t''$ is reached in which $p_{t''} < p$ (or the end of $I'$ is
reached if there is no such step). $t''$ is either due to $c_v>T_v$
or a successful transmission. More precisely, we require that for
any pair of passes $(t,p,S)$ and $(t',p',S')$ with $p'=p$ and final
time step $t''$ in $S$, $(t' \cup S') \cap [t,t''] = \emptyset$, but
passes with $p\not=p'$ are allowed to violate this (by one being
nested into the other). It is not difficult to see that for any
distribution of cumulative probabilities over the time steps of $I'$
one can organize the time steps in $H$ into passes as demanded
above. Based on that, the following claim can be easily shown, where
$k'_1 \le k_1$ is the number of useful time steps with a successful
transmission by a node different from the previously successful
node.

\begin{claim} \label{cl:high1}
\[
  k_0 \ge k'_1 - \log_{1+\gamma} \max\{p_0/(\delta/\varepsilon^2), 1\}
\]
where $p_0$ is the initial cumulative probability in $I'$.
\end{claim}

This is because there can be at most $\log_{1+\gamma}
\max\{p_0/(\delta/\varepsilon^2), 1\}$ many passes not starting with
an idle step but the initial step of $I'$, and every pass has at
most one step counting towards $k'_1$. This also implies the
following claim.

\begin{claim} \label{cl:high1b}
For any collection $P$ of passes,
\[
  k_0 \ge k'_1 - \Delta
\]
where $k_0$ and $k'_1$ are defined w.r.t.~these passes and $\Delta$
is the number of different $p$-values in $P$.
\end{claim}

Also, the following claim holds.

\begin{claim} \label{cl:Tvbound}
\[
  |H| \le (k + \log_{1+\gamma} \max\{p_0/(\delta/\varepsilon^2), 1\}) \sqrt{F}
\]
where $p_0$ is the initial cumulative probability in $I'$.
\end{claim}
\begin{proof}
If at the beginning of $I'$, $T_v \le \sqrt{F}/2$ for every node
$v$, then $T_v \le \sqrt{F}$ for every node $v$ at any time during
$I'$. Hence, after at most $2\sqrt{F}$ non-useful steps we run into
the situation that $c_v>T_v$ for every node $v$, which reduces the
cumulative probability by a factor of $(1+\gamma)$. Given that we
only have $k$ useful steps and we may initially start with a
probability $p_0 > \delta/\varepsilon^2$, there can be at most $(k +
\log_{1+\gamma} \max\{p_0/(\delta/\varepsilon^2), 1\}) \sqrt{F}$
time steps in $H$, which proves the claim.
\end{proof}

For the calculations below recall the definition of $f$ with the
constants $\alpha$ and $\beta$ that are assumed to be sufficiently
large. If $k \le \alpha \log N$, then it follows from
Claim~\ref{cl:Tvbound} that
\[
  |H| \le (\alpha \log N + \log_{1+\gamma} N) \sqrt{F} \le \varepsilon f/\beta
\]
Thus, the number of non-jammed time steps in $H$ is also at most
$\varepsilon f/\beta$, and since $\beta$ can be arbitrarily large,
Lemma~\ref{lem:high} follows.

It remains to consider the case that $k > \alpha \log N$. Let us
assume that $H$ contains at least $\varepsilon f/2$ non-jammed time
steps. Our goal is to contradict that statement in order to show
that the lemma is true. For this we will show that
Claim~\ref{cl:high1b} is violated w.h.p.

Let $T_p$ be the number of all time steps covered by passes
$(t',p',S')$ with $p'=p$. Certainly, $\sum_{p\ge
\delta/\varepsilon^2} T_p = |H|$. Let a pass $(t,p,S)$ be called
{\em bad} if the jamming rate in $S$ is more than
$(1-\varepsilon/8)$. A cumulative probability $p$ is called {\em
bad} if the number of time steps covered by bad passes in $p$ is
more than $(1-\varepsilon/8)T_p$. A bad $p$ contains at least
$(1-\varepsilon/8)^2 T_p$ jammed time steps. Since the number of
jammed time steps in $H$ is at most $|H|-\varepsilon f/2$ it holds
that
\[
  \sum_{p \; {\rm bad}} (1-\varepsilon/8)^2 T_p \le |H|-\varepsilon f/2
\]
Hence, it holds for the good probabilities that
\begin{footnotesize}
\begin{eqnarray*}
  \sum_{p \; {\rm good}} T_p & = & |H| - \sum_{p \; {\rm bad}} T_p \\
  & \ge & |H| - (1-\varepsilon/8)^{-2} (|H| - \varepsilon f/2) \\
  & \ge & f - (1-\varepsilon/8)^{-2} (f - \varepsilon f/2) \ge \varepsilon f/4
\end{eqnarray*}
\end{footnotesize}
In the following, let $\phi =
\delta/\varepsilon^2$ and $\Phi = \ln (f/\log N)$. For each $p \ge
\phi$ let $b_p$ be the number non-idle time steps among the $T_p$
time steps associated with $p$-passes and $k_{0,p}$ be the number of
idle time steps associated with $p$-passes. A good probability $p$
is called {\em helpful} if $b_p \ge k_{0,p} / \Pr[\mbox{idle} \mid
p]$ and $p < \Phi$.

For a cumulative probability $p\ge \Phi$, $\Pr[\mbox{idle} \mid p]
\le e^{-\Phi} = (\log N)/f$ and $\Pr[\mbox{success} \mid p] \le \Phi
e^{-\Phi} \le \ln (f/\log N) \cdot (\log N)/f$. Hence, $k \le \ln f
\cdot \log N$ on expectation, and from the Chernoff bounds it
follows that $k \le 2 \ln f \cdot \log N$ w.h.p., so
Claim~\ref{cl:Tvbound} implies that the number of time steps in $I'$
with cumulative probability $p \ge \Phi$ is at most
\[
  (2 \ln f \cdot \log N + \log_{1+\gamma} N) \sqrt{F} \le \varepsilon f / \beta
\]
If we sum up over all non-helpful probabilities $p$ with $\phi \le p
< \Phi$, they cover at most
\[
  \sum_{i=0}^{\log_{1+\gamma} \Phi} 1/e^{-(1+\gamma)^i}
  \le 2 \cdot f/\log N = o(f)
\]
many time steps, so
\[
  \sum_{p \; {\rm helpful}} T_p \ge \varepsilon f/6
\]
if $\beta$ is large enough. If $\phi \le p_t < \Phi$ and $\Phi \le
1/\gamma$ (which is true if $\gamma = O(1/(\log T + \log \log n))$
is small enough), then it holds for any time step $t'$ with $p_{t'}
\le (1+\gamma) p_t$ that
\begin{footnotesize}
\begin{eqnarray*}
  \lefteqn{ \Pr[\mbox{successful transmission at $t'$}] } \\
  & = & \sum_v p_v(t') \prod_{w \not=v} (1-p_w(t')) \\
  & \ge & \sum_v (1+\gamma)p_v(t) \prod_{w \not=v} (1-(1+\gamma)p_w(t)) \\
  & \ge & \sum_v \frac{(1+\gamma) p_v(t)}{1-\hat{p}} \prod_w (1-(1+\gamma)p_w(t)) \\
  & \ge & \sum_v \frac{(1+\gamma) p_v(t)}{1-\hat{p}} e^{-\sum_w
    (1+\gamma)p_w(t)/(1-(1+\gamma)p_w(t))} \\
  & = & \sum_v \frac{(1+\gamma) p_v(t)}{1-\hat{p}}
    e^{-(1+\gamma)p_t/(1-(1+\gamma)p_t/n)} \\
  & = & \frac{(1+\gamma)p_t}{1-\hat{p}} e^{-(1+\gamma)p_t - 2 p_t^2/n} \\
  & \ge & \frac{(1+\gamma)p_t}{1-\hat{p}} e^{-p_t - 2} \\
  & \ge & \frac{(1+\gamma)\delta}{(1-\hat{p})e^2 \varepsilon} \Pr[\mbox{idle channel at $t$}]
\end{eqnarray*}
\end{footnotesize} Let $k_{1,p}$ be the number of successful time
steps associated with $p$-passes. For each helpful probability $p$
it holds that $\E [ k_{1,p} ]$ is at least
\begin{footnotesize}
\begin{eqnarray*}
   & & \varepsilon/8 \cdot k_{0,p} \cdot
    \frac{\varepsilon/8}{\Pr[\mbox{idle} \mid p]} \cdot \Pr[ \mbox{success} \mid
    (1+\gamma)p ] \\
  & \ge & \varepsilon/8 \cdot k_{0,p} \cdot
    \frac{\varepsilon/8}{\Pr[\mbox{idle} \mid p]} \cdot
    \frac{(1+\gamma) \delta}{(1-\hat{p})e^2 \varepsilon^2} \cdot \Pr[\mbox{idle} \mid
    p] \\
  & \ge & 2
\end{eqnarray*}
\end{footnotesize} if $\delta$ is a sufficiently large constant.
Hence, $\sum_{p \; {\rm helpful}} k_{1,p} \ge 2 \sum_{p \; {\rm
helpful}} k_{0,p}$ and since
\begin{eqnarray*}
  \lefteqn{ \sum_{p \; {\rm helpful}} T_p \cdot \Pr[ \mbox{success} \mid (1+\gamma)p ]} \\
  & \ge & (\varepsilon f/6) \cdot \Phi \cdot e^{-\Phi/(1-\Phi/n)} \\
  & \ge & (\varepsilon f/6) \cdot \Phi \cdot e^{-\Phi - 1} \\
  & = & (\varepsilon f/6) \cdot \ln(f/\log N) \cdot (e \log N)/f \ge c \log N
\end{eqnarray*}
for any constant $c$, the Chernoff bounds imply that $\sum_{p \;
{\rm helpful}} k_{1,p} \ge (3/2) \sum_{p \; {\rm helpful}} k_{0,p}$
w.h.p. In order to proceed, we need the following claim.

\begin{claim} \label{cl:high2}
For any collection $P$ of passes it holds that
\[
  \E[k'_1] \ge (1-(1+\gamma)/n)k_1
\]
where $k_1$ and $k'_1$ are defined w.r.t.~$P$.
\end{claim}
\begin{proof}
Because of Fact~\ref{fa:access3}, the probability that a successful
transmission is done by a node different from the node of the last
successful transmission is equal to
\[
  1- \frac{(1+\gamma)p}{(n + \gamma) p} \ge 1-\frac{1+\gamma}{n}.
\]
To see this, observe that among the cumulative probability $p$, if
the last sender $u$ has a share $p_u(t)=x$, all other nodes $v$ have
a share $x/(1+\gamma)$, and
\[
  \frac{p_u(t)}{\sum_{v\in V} p_v(t)} = \frac{x}{(n-1)\cdot
  \frac{x}{1+\gamma}+x} = \frac{1+\gamma}{n+\gamma}
\]
Hence, $\E[k'_1] \ge (1-(1+\gamma)/n)k_1$.
\end{proof}

The claim implies that
\begin{footnotesize}
\[
  \sum_{p \; {\rm helpful}} k'_{1,p} > \sum_{p \; {\rm helpful}} k_{0,p} +
  \log_{1+\gamma} \max\{\Phi/(\delta/\varepsilon^2), 1\}
\]
\end{footnotesize}
w.h.p., which violates Claim~\ref{cl:high1b}. This
completes the proof of Lemma~\ref{lem:high}.
\end{proof}

Notice that by the choice of $f$ and $F$, $T_v$ never exceeds
$\sqrt{F}/2$ for any $v$ when initially $T_v=1$ for all $v$. Hence,
the prerequisites of the lemmas are satisfied. We can also show the
following lemma, which shows that $T_v$ remains bounded over time.

\begin{lemma}\label{lemma:tv}
For any time frame $I$ in which initially $T_v \le \sqrt{F}/2$ for
all $v$, also $T_v \le \sqrt{F}/2$ for all $v$ at the end of $I$
w.h.p.
\end{lemma}
\begin{proof}
We already know that in each subframe $I'$ in $I$, at least
$\varepsilon f/2$ of the non-jammed time steps $t$ in $I'$ satisfy
$p_t \le \delta/\varepsilon^2$ w.h.p. Hence, for all
$(T,1-\varepsilon)$-bounded jamming strategies, there are at least
\[
  (\delta/\varepsilon^2) \cdot e^{-\delta/\varepsilon^2} \cdot \varepsilon f/2
\]
useful time steps in $I'$ w.h.p. Due to the lower bound of $p_t \ge
1/(f^2 (1+\gamma)^{\sqrt{f}})$ for all time steps in $I$ w.h.p.~we
can also conclude that
\[
  k_0 \ge k'_1 + k_2 - \log_{1+\gamma} ((\delta/\varepsilon^2) \cdot f^2
  (1+\gamma)^{\sqrt{f}})
\]
Because of Claims~\ref{cl_Tinc} and \ref{cl:high2} it follows that
\[
  k_0 \ge k_1/3
\]
w.h.p. Since $k_0+k_1 = k$ and $k \ge (\delta/\varepsilon^2) \cdot
e^{-\delta/\varepsilon^2} \cdot \varepsilon f/2$ it follows that
$k_0 = \Omega(f)$. Therefore, there must be at least one time point
in $I'$ with $T_v=1$ for all $v \in V$. This in turn ensures that
$T_v \le \sqrt{F}/2$ for all $v$ at the end of $I$ w.h.p.
\end{proof}

With Lemma~\ref{lemma:tv}, we show that Lemma~\ref{lem:high} is true
for a polynomial number of subframes. Then, Lemma~\ref{lem:high} and
Lemma~\ref{lemma:tv} together imply that
Lemma~\ref{lem_enough_useful} holds for a polynomial number of
subframes. Hence, our main Theorem~\ref{th_main} follows. Along the
same line as in~\cite{singlehop08}, we can show that
\textsc{AntiJam} is self-stabilizing, so the throughput result can
be extended to an arbitrary sequence of time frames.

\begin{figure*} [t]
\begin{center}
\includegraphics[width=0.4\textwidth]{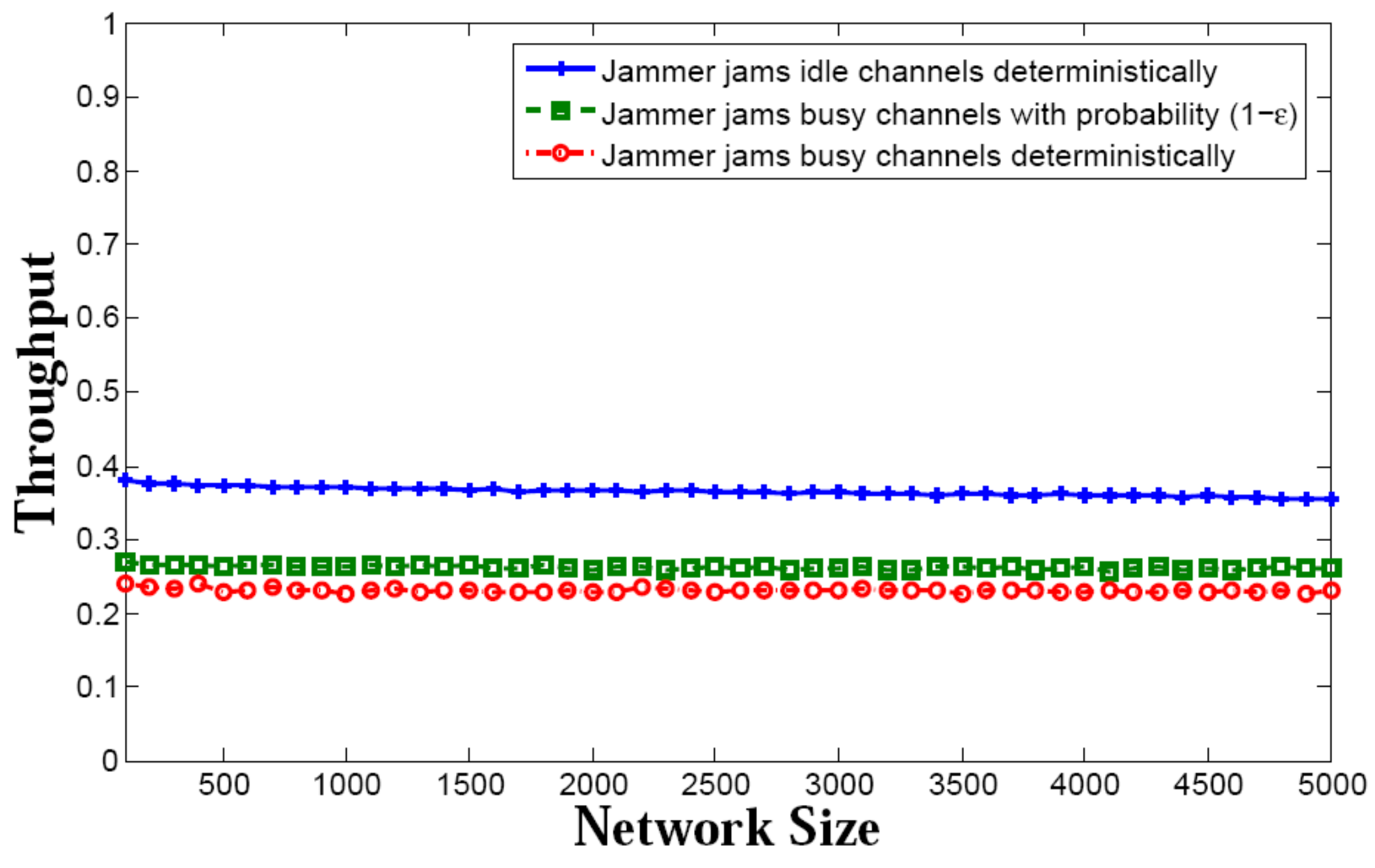}~\includegraphics[width=0.4\textwidth]
{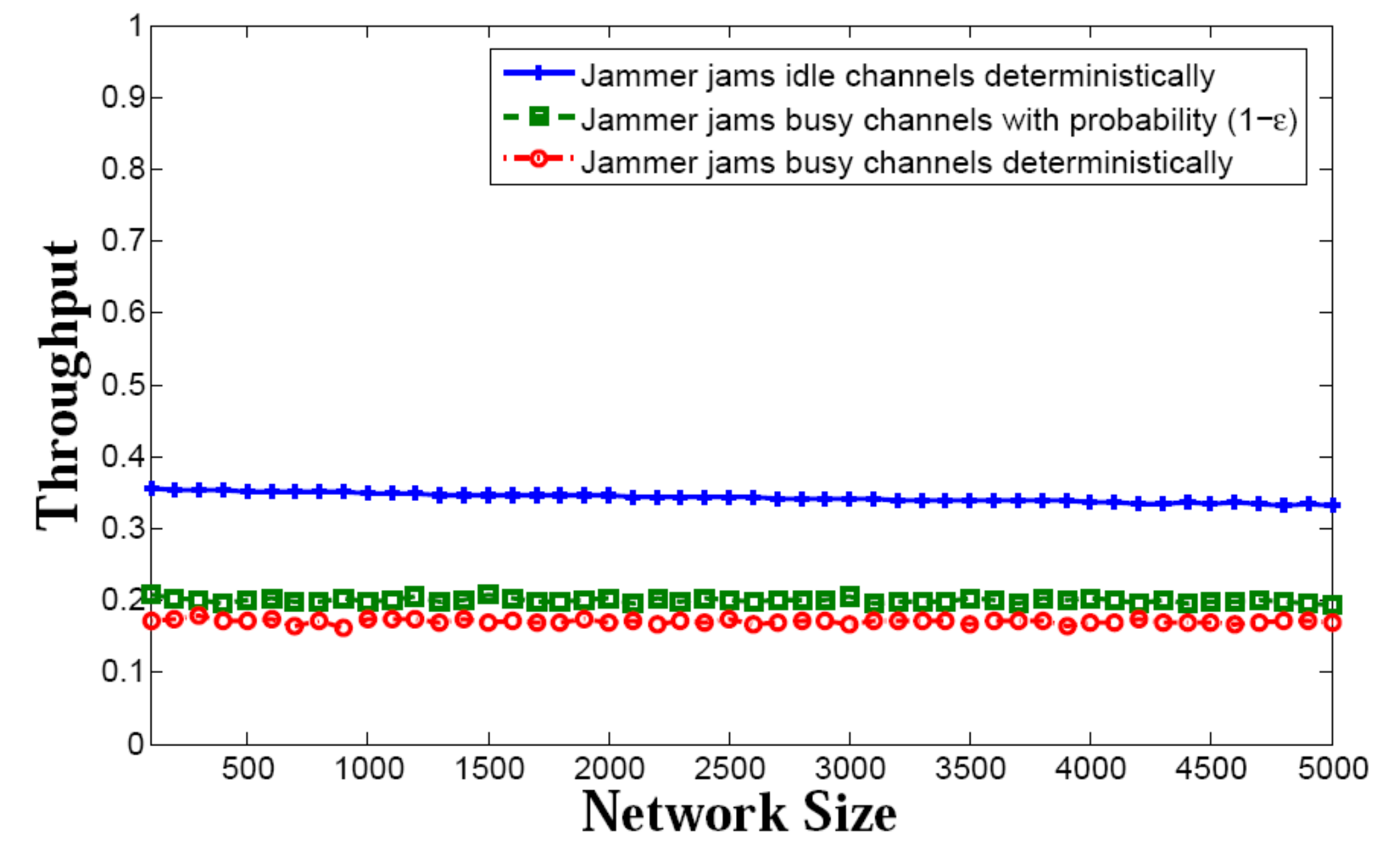}\\
\caption{Throughput under three different jamming strategies as a
function of the network size and $\varepsilon$, where $\hat{p}=1/24$
(\emph{left:} $\varepsilon=0.5$, \emph{right:}
$\varepsilon=0.3$).}\label{fig:jamming-strategies-comp}
\end{center}
\end{figure*}

\begin{figure*} [t]
\begin{center}
\includegraphics[width=0.4\textwidth]{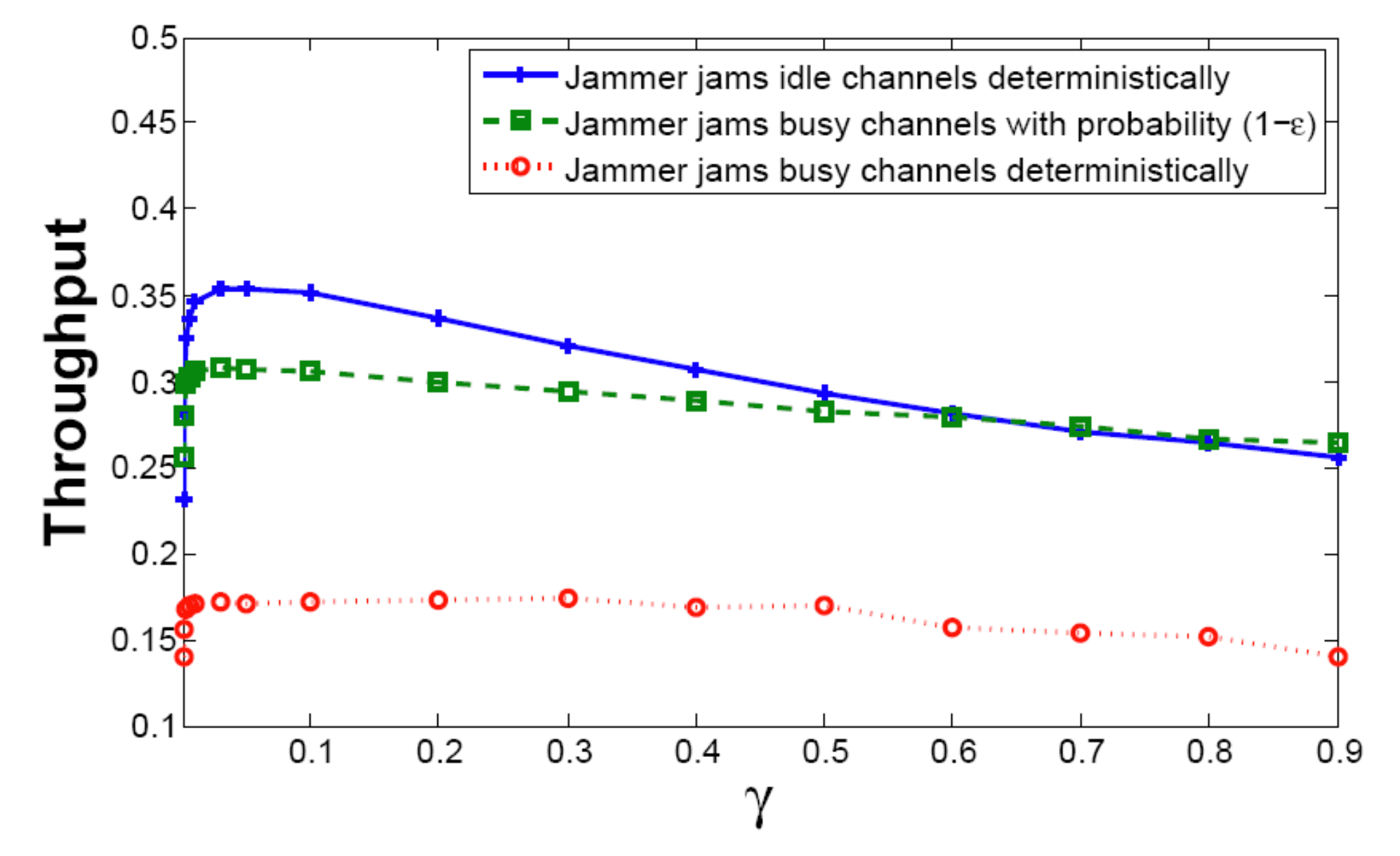}~\includegraphics[width=0.4\textwidth]
{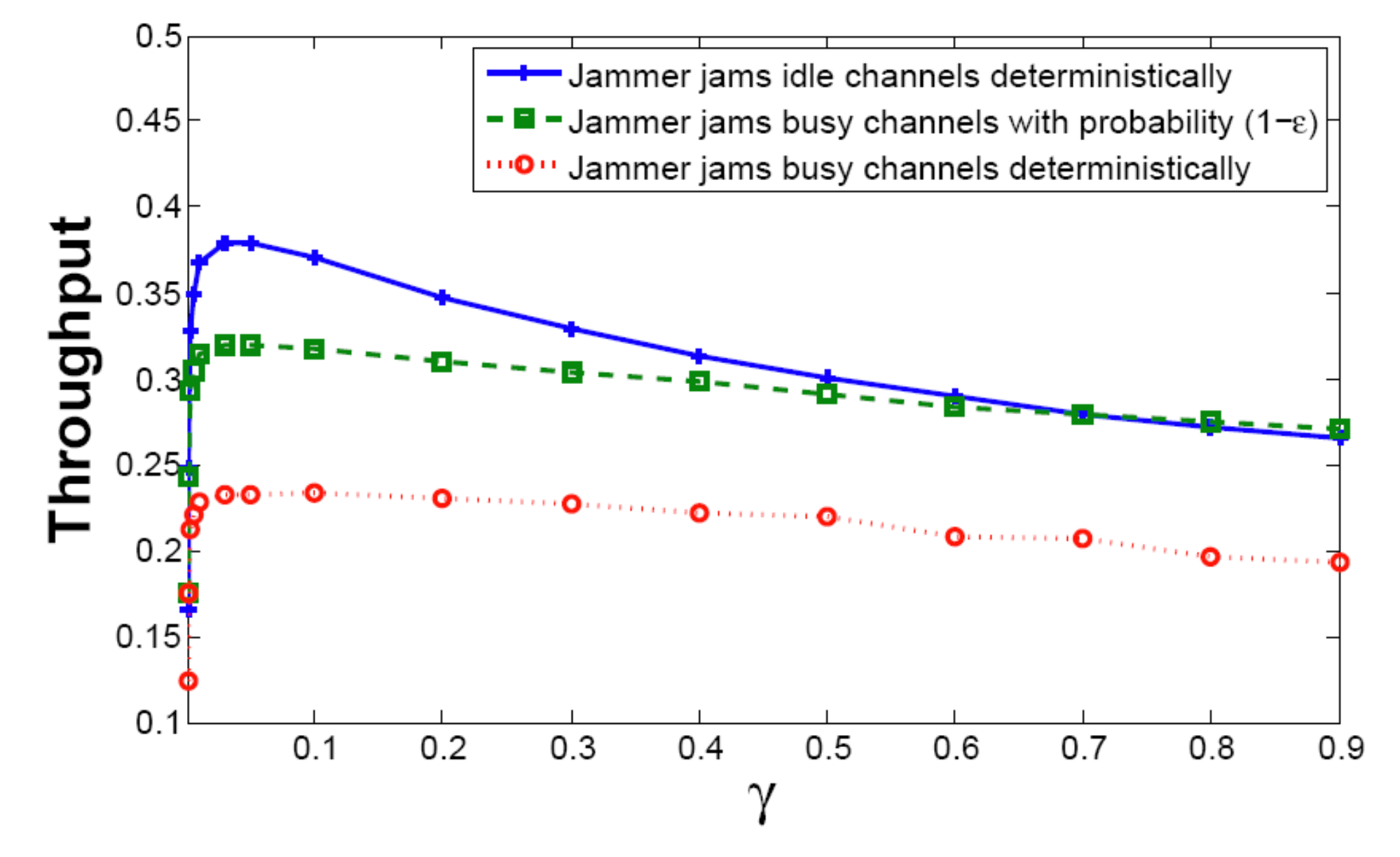}\\
\caption{Throughput as a function of $\gamma$ under three different
jamming strategies. \emph{Left:} $\varepsilon=0.2$, \emph{Right:}
$\varepsilon=0.5$).}\label{fig:gamma}
\end{center}
\end{figure*}

\section{Simulation}\label{sec:experiments}

We have implemented a simulator to study additional properties of
our protocol. This section reports on some of our results. Our focus
here is on the qualitative nature of the performance of
\textsc{AntiJam}, and we did not optimize the parameters to obtain
the best constants. We consider three different jamming strategies
for a reactive jammer that is $(T,1-\varepsilon)$-bounded, for
different $\varepsilon$ values and where $T = 100$: (1) one that
jams busy channels with probability $(1-\varepsilon)$;
  (2) one that jams busy channels deterministically (as long the jamming budget is not used up);
    (3) one that jams idle channels deterministically (as long as the jamming budges is not used up).

We define throughput as the number of successful transmissions over
the number of non-jammed time steps.

\subsection{Throughput}

In a first set of experiments we study the throughput as a function
of the network size and $\varepsilon$. We evaluate the throughput
performance for each type of adversary introduced above, see
Figure~\ref{fig:jamming-strategies-comp}. For all three strategies,
the throughput is basically constant, independently of the network
size; this is in accordance with our theoretical insight of
Theorem~\ref{th_main}. We can see that given our conditions on
$\varepsilon$ and $T$, the strategy that jams busy channels
deterministically results in the lowest throughput. Hence, in the
remaining experiments described in this section, we will focus on
this particular strategy. As expected, jamming idle channels does
not affect the protocol behavior much.

In our simulations, \textsc{AntiJam} makes effective use of the
non-jammed time periods, yielding 20\%-40\% successful transmissions
even without optimizing the protocol parameters. In additional
experiments we also studied the throughput as a function of
$\gamma$, see Figure~\ref{fig:gamma}. As expected, the throughput
declines slightly for large $\gamma$, but this effect is small.
(Note that for very small $\gamma$, the convergence time becomes
large and the experiments need run for a long time in order not to
underestimate the real throughput.)

\subsection{Convergence Time}

Besides a high throughput, fast convergence is the most important
performance criterion of a MAC protocol. The traces in
Figure~\ref{fig:convergence} show the evolution of the cumulative
probability over time. It can be seen that the protocol converges
quickly to constant access probabilities. (Note the logarithmic
scale.) If the initial probability for each node is high, the
protocol needs more time to bring down the low-constant cumulative
probability. Moreover, the ratio of the time period the cumulative
probability is in the range of $[\frac{1}{2\varepsilon},
\frac{2}{\varepsilon}]$ to the time period the protocol being
executed is $92.98\%$ when $\hat{p} = 1/24$, and $89.52\%$ when
$\hat{p} = 1/2$. This implies that for a sufficiently large time
period, the cumulative probability is well bounded most of the time,
which corresponds to our theoretical insights.
Figure~\ref{fig:convergence-over-size} studies  the convergence time
for different network sizes. We ran the protocol 50 times, and
assume that the execution has converged when the cumulative
probability $p$ satisfies $p\in[0.1, 10]$, for at least $5$
consecutive rounds. The simulation result also confirms our
theoretical analysis in Theorem~\ref{th_main}, as the number of
rounds needed to converge the execution is bounded by
$\Theta(\frac{1}{\varepsilon} \log N \max\{T, \frac{1}{\varepsilon
\gamma^2} \log^3 N\})$.

\begin{figure} [h]
\begin{center}
\includegraphics[width=0.9\columnwidth]{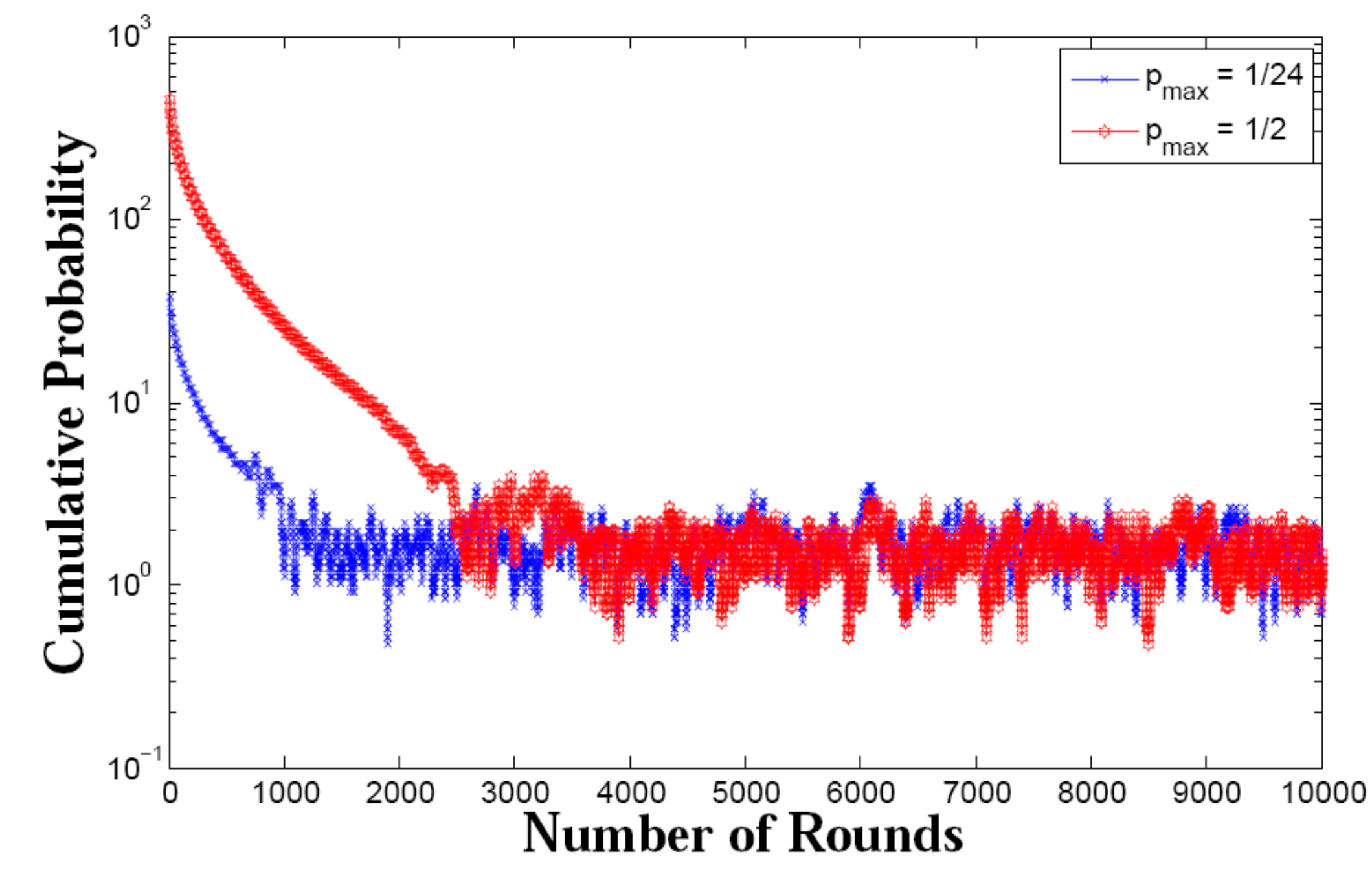}\\
\caption{Evolution of cumulative probability over time (network size
is 1000 nodes, and $\varepsilon = 0.5$). Note that the plot has
logarithmic scale.}\label{fig:convergence}
\end{center}
\end{figure}

\begin{figure} [h]
\begin{center}
\includegraphics[width=0.9\columnwidth]{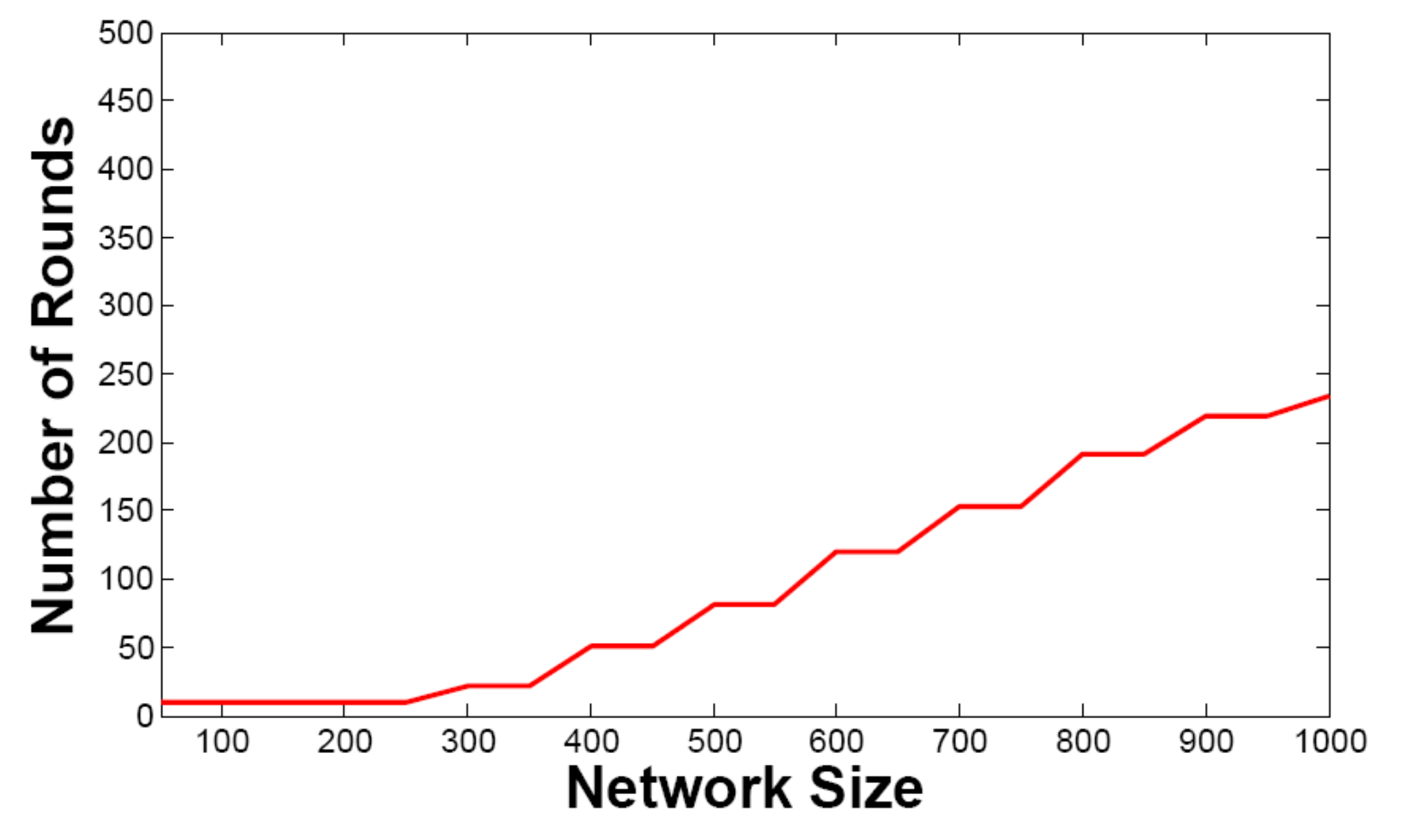}\\
\caption{\textsc{AntiJam} runtime as a function of network size for
$\hat{p}=1/24$, and $\varepsilon =
0.5$.}\label{fig:convergence-over-size}
\end{center}
\end{figure}

Figure~\ref{fig:self-stab} indicates that independently of the
initial values $\hat{p}$ and $T_v$, the throughput rises quickly (up
above $20\%$) and stays there afterwards.

\begin{figure} [h]
\begin{center}
\includegraphics[width=0.9\columnwidth]{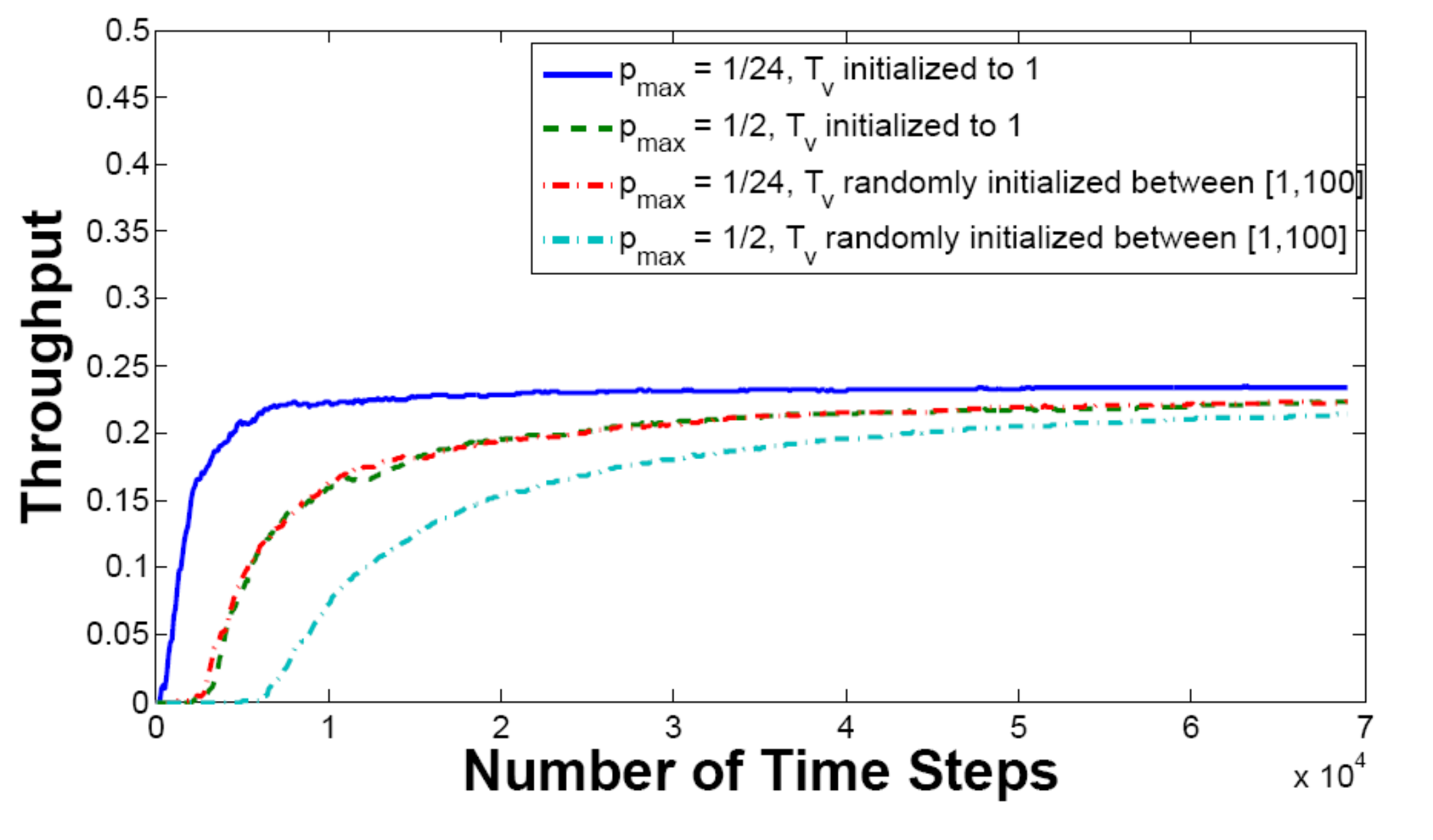}\\
\caption{Convergence in a network of 1000 nodes where $\varepsilon =
0.5$.}\label{fig:self-stab}
\end{center}
\end{figure}

\subsection{Fairness}

As \textsc{AntiJam} synchronizes $c_v$, $T_v$, and $p_v$ values upon
message reception, the nodes are expected to transmit roughly the
same amount of messages; in other words, our protocol is fair.
Figure~\ref{fig:fairness} presents a histogram showing how the
successful transmissions are distributed among the nodes. More
specifically, we partition the number of successful transmissions
into intervals of size $4$. Then, all the transmissions are grouped
according to those intervals in the histogram.
\begin{figure} [h]
\begin{center}
\includegraphics[width=0.9\columnwidth]{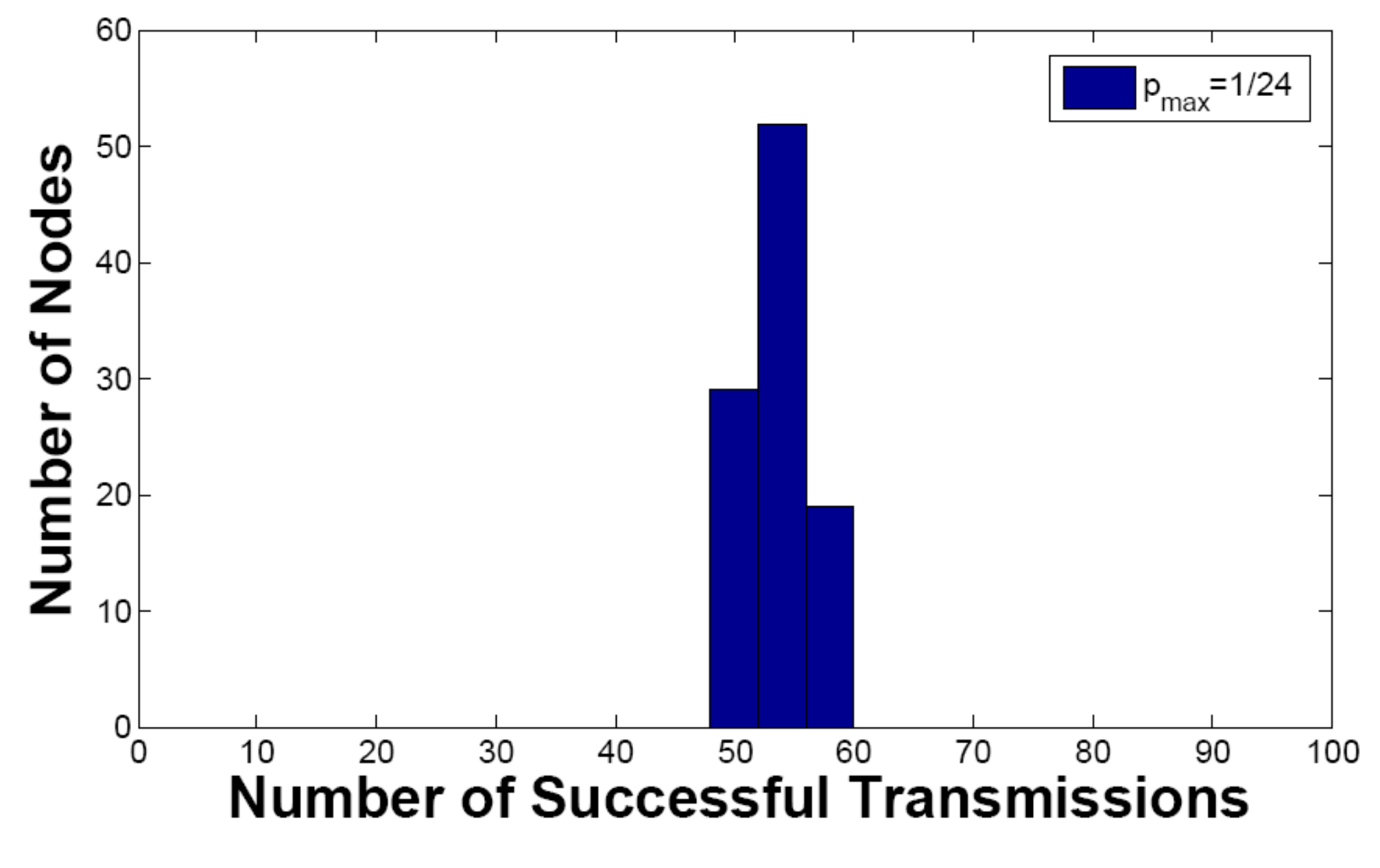}\\
\caption{Fairness in a network of 1000 nodes, where $\varepsilon =
0.5$, and $\hat{p}=1/24$ (averaged over 10
runs).}\label{fig:fairness}
\end{center}
\end{figure}

\subsection{Comparison to 802.11}

Finally, to put \textsc{AntiJam} into perspective, as a comparison,
we implemented a simplified version of the widely used 802.11 MAC
protocol (with a focus on 802.11a).

The configurations for the simulation are the following: (1) the
jammer is reactive and $(T,1-\varepsilon)$-bounded; (2) the unit
slot time for 802.11 is set to $50\mu s$; for simplicity, we define
one time step for \textsc{AntiJam} to be $50\mu s$ also; (3) we run
\textsc{AntiJam} and 802.11 for 4 min, which is equal to $4.8\cdot
10^6$ time steps in our simulation; (4) the backoff timer of the
802.11 MAC protocol implemented here uses units of $50\mu s$; (5) we
omit SIFS, DIFS, and RTS/CTS/ACK.

A comparison is summarized in Figure~\ref{fig:jadeplus-vs-others}.
The throughput achieved by \textsc{AntiJam} is significantly higher
than the one by the 802.11 MAC protocol, specially for lower values
of $\varepsilon$, when the 802.11 MAC protocol basically fails to
deliver any successful message.

\begin{figure} [h]
\begin{center}
\includegraphics[width=0.9\columnwidth]{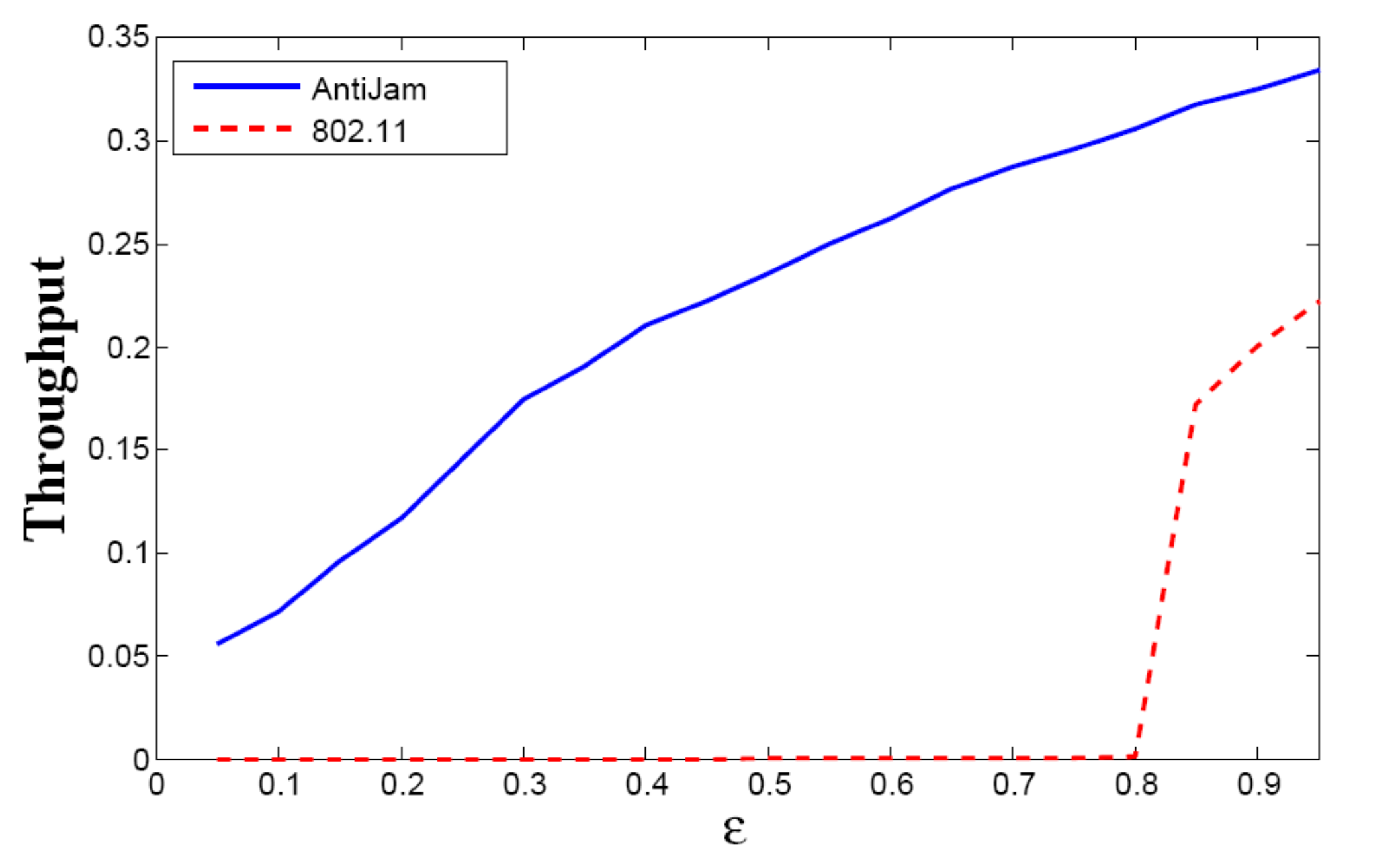}\caption{Throughput as a function of
$\varepsilon \in [0.05, 0.95]$, compared to 802.11, averaged over 10
runs, where $\hat{p} = 1/24$.}\label{fig:jadeplus-vs-others}
\end{center}
\end{figure}

\section{Conclusion}\label{sec:conclusion}

This article presented a simple distributed MAC protocol called
\textsc{AntiJam} that is able to make efficient use of a shared
communication medium whose availability is changing quickly and in a
hard to predict manner over time. In particular, this article has
shown that the MAC protocol is able to achieve a good
(asymptotically optimal) throughput even against an adaptive and
reactive jammer that uses carrier sensing for an informed decision
on when to jam, and whose strategy can depend on the entire protocol
history. Our simulation results indicate that the nodes' access
probabilities converge quickly to a good cumulative value and yields
a fair allocation of the shared medium among the nodes.

{\footnotesize \bibliographystyle{abbrv}
\bibliography{icdcs11}
}

\end{document}